\documentclass[11pt]{article}
\usepackage{styles}
\usepackage{physics}
\newcommand{\negl}{\mathrm{negl}}
\newcommand{\Exec}{\mathsf{Exec}}
\newcommand{\out}{\mathsf{out}}
\newcommand{\secp}{\lambda}
\newcommand{\win}{\mathsf{win}}
\newcommand{\inp}{\mathsf{in}}
\newcommand{\claw}{\mathsf{claw}}
\newcommand{\bbstate}{\mathsf{BB84}}
\newcommand{\blind}{\mathsf{blind}}

\title{Succinct Arguments for QMA from Collapsing Hash Functions}
\date{}
\author{James Bartusek\thanks{Columbia University \href{mailto:bartusek.james@gmail.com}{bartusek.james@gmail.com}} \and Giulio Malavolta\thanks{Bocconi University \href{mailto:giulio.malavolta@unibocconi.it}{giulio.malavolta@unibocconi.it}}}

\begin{document}

\maketitle
\begin{abstract}
    We prove the existence of succinct arguments for QMA, assuming only the existence of collapsing hash functions. This is the first scheme that relies only on unstructured ``Minicrypt'' assumptions, which are not known to imply public-key encryption. 

    Our main technical contribution is a quantum-succinct \emph{claw-state generation} protocol that allows us to bootstrap a small number of quantum correlations into an arbitrarily large number of claw-state correlations, using classical communication only. This improves upon the work of [Zhang, STOC 2021], having better round complexity, a proof in the standard model, and being overall much simpler. This yields a quantum-succinct blind delegation of quantum computation protocol from one-way functions, which we plug into the communication-compression compiler of [Bartusek, Liu, and Malavolta, EUROCRYPT 2026] to obtain succinct arguments for QMA.
\end{abstract}

\begingroup
\setlength{\parskip}{0pt}
\tableofcontents
\endgroup
\clearpage

\section{Introduction}

Succinct arguments enable the verification of mathematical statements using significantly fewer resources than required to process a complete proof. Pioneered by Kilian \cite{Kil92}, succinct arguments have become a pillar of study in both foundational and applied cryptography. In his work, Kilian showed how to compile any probabilistically checkable proof system (PCP) \cite{ALM+98,AS98} for NP into a succinct argument for NP using only collision-resistant hash functions.

Recent years have witnessed a surge of interest in quantum information-processing, giving rise to a natural follow-up question: Under what cryptographic assumptions do there exist succinct arguments for all of QMA, namely, what type of cryptographic structure enables the extremely efficient verification of statements with \emph{quantum} proofs? A series of recent works \cite{BKL+22,MNZ24,GTKNV25,BLM26} has culminated in a construction of succinct arguments for QMA from the following two cryptographic ingredients:
\begin{itemize}
    \item Collapsing hash functions: the post-quantum analogue of collision-resistance.
    \item Oblivious state preparation: a generic ``public-key'' style assumption that can be instantiated from LWE, LIP, or assumptions on cryptographic group actions \cite{BK25,BMM25}. 
\end{itemize}

The state of the art thus highlights a significant gap from the classical setting: While succinct arguments for NP only require unstructured cryptography in the form of collision-resistant hash functions (placing them in ``Minicrypt''), succinct arguments for QMA are only known from public-key style assumptions, positioning them in ``Cryptomania''.

In this work, we establish that succinct arguments for QMA do in fact live in Minicrypt by proving their existence assuming only collapsing hash functions \cite{Unr16}. As collapsing is the post-quantum analogue of collision-resistance, this exactly matches the weakest assumption under which Kilian's succinct arguments for NP are known to be post-quantum secure \cite{CMSZ22}.

\subsection{Results}

Our main result is stated below, where $\lambda$ denotes the security parameter. 

\begin{theorem}\label{thm:intro}
There exists a \emph{fixed} polynomial $\poly$ such that, assuming collapsing hash functions, there exists an argument system for QMA with completeness $1-\negl(\lambda)$, soundness error $\negl(\lambda)$, total communication bounded by $\mathrm{poly}(\lambda)$, and verifier runtime $|x|\cdot\mathrm{poly}(\lambda)$.

The protocol uses quantum communication (and thus a quantum verifier) and makes non-black-box use of the collapsing hash function.
\end{theorem}
Above, $x$ refers to the QMA statement, and $|x|$ is its bit length. In particular, $|x|\cdot \poly(\lambda)$ has no dependence on the witness length or the time required to run the original QMA verifier on $x$ and a witness state $\ket{\psi}$. Note that, because of the non-black-box use of the collapsing hash function, we do not automatically obtain a succinct argument for QMA in the quantum random oracle model.

\begin{remark}[Prepare-and-send verifier]
One potentially desirable aspect of our protocol is that it does not require the verifier to keep any quantum memory between rounds. In particular, each of the verifier's quantum messages can be prepared from its classical state immediately before sending, and no auxiliary quantum registers are kept by the verifier after the message is sent. 
\end{remark}

Our main technical ingredient is a new construction of blind delegation of quantum computation, which asks whether a resource-constrained client can delegate a quantum computation to a server in such a manner that the server learns nothing about the client's computation. One natural resource constraint is the size of the quantum circuit that the client runs during the course of the protocol. If the size is 0, we refer to this as \emph{classical} blind delegation of quantum computation, and if the size is bounded by a fixed polynomial in the security parameter, \emph{independent} of the size of the computation being delegated, we refer to this as \emph{quantum-succinct} blind delegation of quantum computation.

It was shown by \cite{Zha21} that quantum-succinct blind delegation of quantum computation exists in the quantum random oracle model. In this work, we design and prove the security of a vastly simpler protocol that only requires one-way functions, yielding the following theorem.

\begin{theorem}\label{thm:intro-blind}
    Assuming one-way functions, there exists quantum-succinct blind delegation of quantum computation. 
\end{theorem}

\paragraph{Discussion.} Recall that Kilian's result used collision-resistant hash functions to compile a PCP for NP into a succinct argument for NP. Recently, a quantum analogue of Kilian's compiler was worked out by \cite{GJMZ23}, establishing that PCPs for QMA can be compiled into succinct arguments for QMA, utilizing only unstructured cryptography. In fact, they assume only the existence of pseudorandom unitaries, a ``Microcrypt'' assumption that is weaker than even one-way functions. Unfortunately, it remains a major unresolved problem whether PCPs for QMA actually exist. 

Another point worth making is that \cite{GJMZ23} and our protocol both make use of quantum communication and thus a quantum verifier. One could ask whether the verifier can be made completely classical, while remaining succinct. This is indeed known, but under public-key assumptions such as LWE. However, improving the result to classical, even \emph{non-succinct}, verification of QMA in Minicrypt would be considered a major breakthrough, so we regard the quantum communication component of our protocol as a crucial relaxation given current techniques.

Therefore, we can summarize the state of succinct arguments for QMA as follows.

\begin{itemize}
    \item \cite{GJMZ23}: Assuming the quantum PCP conjecture, there exist succinct arguments for QMA from pseudorandom unitaries. That is, assuming the quantum PCP conjecture, succinct arguments for QMA exist in Microcrypt.
    \item Our result: Succinct arguments for QMA exist in Minicrypt, without assuming the quantum PCP conjecture. In particular, succinct arguments for QMA exist assuming only collapsing hash functions.
    \item \cite{MNZ24,GTKNV25,BLM26}: Succinct arguments for QMA \emph{with classical verification} exist in Cryptomania, without assuming the quantum PCP conjecture. In particular, succinct classically-verifiable arguments for QMA exist assuming a standard public-key assumption such as LWE.
\end{itemize}

\subsection{Technical Outline}

Before we explain our approach to the problem, let us recall some necessary background information to motivate our design choices.
Our starting point is the recent work of~\cite{BLM26}, who proposed a compiler that, assuming collapsing hash functions, compresses the communication complexity of any classical-verifier interactive protocol $\Pi$ with $r(\lambda)$ rounds. For any $\varepsilon=1/\poly(\lambda)$, the compiled protocol $\tilde\Pi$ has negligible completeness loss, soundness loss $O(\varepsilon)+\negl(\lambda)$, and satisfies
\[
 \mathrm{CC}(\widetilde\Pi)\le r(\lambda)\poly(\lambda,1/\varepsilon),
 \qquad
 \text{verifier time}\le |x|\,r(\lambda)\poly(\lambda,1/\varepsilon),
\]
where $x$ is the common input and the polynomial is fixed independently of the original message lengths. Thus, for our purposes, the remaining task is to construct a QMA argument with a fixed $\poly(\lambda)$ number of rounds but otherwise arbitrary classical communication, which can be fed into the above compiler.

Prior work \cite{BLM26} constructs such a round-efficient argument using the compiled non-local games framework~\cite{KLVY23,NZ23,BK25,BKMSW25}. The approach here is to start from an information-theoretically secure \emph{two-prover} protocol and compile it into a single-prover protocol using cryptography. The cryptographic ingredient is \emph{blind delegation}, which is an interactive protocol between a server and a client, where the client holds a private classical input $x$, the server has a quantum input $\rho$, and they both have the public description of a quantum circuit $Q$. In the ideal honest execution, the client obtains Pauli keys $(r,s)$ and the server obtains $\sigma$ such that
\[
 \mathsf{Z}^s\mathsf{X}^r\sigma\mathsf{X}^r\mathsf{Z}^s=Q(x,\rho).
\]
Correctness allows negligible error and includes arbitrary reference systems, as specified in \Cref{def:blind}. We require that, for every QPT server $\mathcal A$ and every pair of equal-length inputs $(x_0, x_1)$, the views of the server are computationally indistinguishable:
\[
\mathsf{View}_{\mathcal A}(\rho,x_0)\approx_c\mathsf{View}_{\mathcal A}(\rho,x_1).
\]
\cite{BLM26} shows how to combine a blind delegation protocol from \cite{BK25,BKMSW25} with a two-prover game from \cite{MNZ24} to obtain the desired round-efficient argument. These protocols however require the existence of public-key assumptions such as oblivious state preparation or trapdoor claw-free functions. Constructing blind delegation from collapsing hashes alone appears unlikely (or at least very difficult given known techniques), since blind delegation implies (classical-communication) key agreement \cite{BK25}.

\paragraph{Idea I: Relaxing the Model.}
Our first simple, but crucial, observation is that the client does not need to be completely classical in the above interaction. For instance, it could be acceptable to construct a blind delegation protocol where the client performs a \emph{fixed polynomial amount} of quantum operations, so long as we can still apply the aforementioned communication-compression compiler to the (potentially long) classical messages. Henceforth, we refer to such a protocol as \emph{quantum succinct}. With this observation in mind, let us try to simplify the problem even further: In \cite{BK25} it is shown that a blind delegation protocol can be generically derived from \emph{claw-state generation}, an interactive protocol where at the end, the server obtains multiple claw states
\[
\frac{1}{\sqrt{2}}(\ket{0,x_0}+\ket{1,x_1})
\]
and the client obtains a classical description of their labels $(x_0, x_1)$. The only security guarantee that we require is that the server should not be able to guess \emph{both} $x_0$ and $x_1$. Thus, our problem further reduces to designing a round-efficient and quantum-succinct claw-state generation protocol: Allowing the client to use $\poly(\lambda,\log T)$ quantum gates and a (possibly large) amount of classical communication, we want an interactive protocol where the server obtains $T$ such claw states. Climbing up the chain of implications explained above, this will suffice for our main result.

In fact, a very similar problem was considered by Zhang~\cite{Zha21}, precisely motivated by the construction of a blind delegation protocol with a quantum-succinct client. He shows how to \emph{bootstrap} a small initial quantum communication into an arbitrary polynomial number of claw states, using classical communication only. However his approach has two limitations: First, the round complexity of his protocol is polynomial in $T$, rendering it unusable for us since we require a round-efficient protocol. Second, the security is proved in the quantum random-oracle model. This is also problematic, since the communication-compression compiler makes non-black-box use of the underlying protocol (in this case, it would require the circuit description of the underlying hash function). This would render the security claim of the underlying protocol vacuous.

Instead, we propose a new approach to quantum-succinct claw-state generation, which relies only on (quantum secure) one-way functions and has the added benefit of being dramatically simpler than Zhang’s protocol. Next, we give a more detailed overview of our protocol.

\paragraph{Idea II: Quantum-Succinct Claw-Generation from One-Way Functions.}
Assume that $T$ is a power of $2$ and set $n=\lambda$. Prior to the interaction, the client samples several PRF keys denoted by $k\in\{0,1\}^{\lambda}$ and $k_{i,b}\in\{0,1\}^{\lambda}$, along with random strings $x_{i,b}\in\{0,1\}^n$, for $i\in[\log T]$ and $b\in\{0,1\}$. For any index $t \in [T] \cong\{0,1\}^{\log T}$ we will think of $(y_{t,0},y_{t,1})=f_k(t)$ as the strings that determine the $t$-th output claw. The client prepares and sends the quantum state 
\begin{equation}\label{eq:sent}
 \frac1{\sqrt T}\sum_{t\in\{0,1\}^{\log T}}
 \left(\bigotimes_{i\in[\log T]}\ket{t_i,x_{i,t_i},k_{i,t_i}}\right)
 \frac{\ket{0,y_{t,0}}+\ket{1,y_{t,1}}}{\sqrt2}
\end{equation}
to the server, which can be prepared efficiently (that is, in time $\poly(\lambda, \log T)$) with the knowledge of the client's randomness. In addition the client also sends the large, but classical, table
\begin{equation}\label{eq:table}
 c_{t,i}=f_k(t)\oplus f_{k_{i,1-t_i}}(t)
 \qquad \text{ for } t\in\{0,1\}^{\log T},\ i\in[\log T].
\end{equation}
To get some intuition, it is useful to think about each of the $T$ states in the superposition of \Cref{eq:sent}: In the ``$t$-subspace'', the server can compute $(y_{t',0},y_{t',1})$ for every $t'$ \emph{except for $t$}. Indeed, for $t'\ne t$, choose any $i$ with $t'_i\ne t_i$. Then
$c_{t',i}\oplus f_{k_{i,t_i}}(t')=f_k(t') = (y_{t',0},y_{t',1})$. We also mention that while the random $x_{i,t_i}$ are not used in this derivation, their presence will be useful for arguing security.

Coherently preparing the missing claws and placing every output at its designated index, the server can then prepare the state
\begin{equation}\label{eq:expanded}
 \left(\bigotimes_{i\in[\log T]}
 \frac{\ket{0,x_{i,0},k_{i,0}}+\ket{1,x_{i,1},k_{i,1}}}{\sqrt2}\right)
 \otimes\bigotimes_{t\in\{0,1\}^{\log T}}\frac{\ket{0,y_{t,0}}+\ket{1,y_{t,1}}}{\sqrt2}.
\end{equation}
To complete the protocol, the server measures the first registers in the Hadamard basis and returns $d=(d_1,\ldots,d_{\log T})$, parsed as $d_i=(d_i^{(0)},d_i^{(1)},d_i^{(2)})\in\{0,1\}\times\{0,1\}^n\times\{0,1\}^{\lambda}$. The client accepts exactly when, for every $i$,
\begin{equation}\label{eq:checks}
 d_i^{(1)}\ne0,
 \qquad
 \langle d_i,(1,x_{i,0}\oplus x_{i,1},k_{i,0}\oplus k_{i,1})\rangle=0.
\end{equation}
Honest Hadamard measurements satisfy the parity constraint, and $d_i^{(1)}=0$ occurs with negligible probability. The honest server thereby deletes the first registers and retains the output claw states. Moreover, notice that the round complexity of this protocol is constant, and the client's quantum operations are bounded by $\poly(\log T,n,\lambda) = \poly(\lambda)$, whereas the client's classical runtime and the server's total runtime are $\poly(T,\lambda)$.


For security, we will only be able to prove that, for a \emph{random} choice of index $t^*\gets[T]$ sampled at the conclusion of the protocol, the adversary cannot both pass \Cref{eq:checks} and output $f_k(t^*) = (y_{t^*,0},y_{t^*,1})$ with better than \emph{constant} probability. We refer the reader to the technical sections for a precise analysis, and here we just say that the proof will boil down to splitting the state of the adversary into two components:
\begin{itemize}
    \item The first component will be bounded by the success probability of an experiment where the adversary is given \emph{only} the part of the state indexed by $t^*$, instead of the superposition in \Cref{eq:sent}, which allows us to appeal to the security of the PRF and replace the table entries $c_{t^*,i}$ by uniform strings and $f_k(t^*)$ by a uniform pair. In this hybrid experiment, the success probability of outputting $f_k(t^*)$ can therefore be bounded by negligible. 
    \item The second component is a residual ``unbalanced'' state derived by subtracting off the (appropriately weighted) $t^*$ branch, and whose averaged squared norm after the deletion check can only be at most $\frac14$: each individual branch passes the checks with probability at most $\frac1T$, and the imbalance in the amplitudes of the branches can then be used to derive the $\frac{1}{4}$ bound.
\end{itemize}

Overall, we combine the two bounds using standard inequalities, and obtain a total bound on passing \Cref{eq:checks} and guessing the randomly selected output pair by $\frac14$ (plus an irrelevant negligible term). The next goal is therefore to \emph{amplify} the hardness of these output claw states.

\paragraph{Hardness Amplification.}
To amplify, we use an idea from Zhang \cite{Zha21}, which consists in running $\lambda$ independent copies of the expansion protocol sequentially and then combining $\lambda$ weakly-secure claw states into one strongly-secure claw state. Only after all copies terminate, the client chooses independent uniform permutations $\pi_i\in S_T$ for each repetition $i \in [\lambda]$ and sends them to the server. 
For an output index $t$, the server combines the $\lambda$ claws indexed by $\{\pi_i(t)\}_{i \in [\lambda]}$ by measuring the XOR of the first claw's leading qubit with each other leading qubit, obtaining the outcomes $e_{2,t},\dots,e_{\lambda,t}$. This effectively \emph{glues} together independent claw states into a single claw $\frac{\ket{0,Y_{t,0}}+\ket{1,Y_{t,1}}}{\sqrt2}$, where
\begin{align*}
 Y_{t,0}&=y^{(1)}_{\pi_1(t),0}\mathbin\Vert
 y^{(2)}_{\pi_2(t),e_{2,t}}\mathbin\Vert\cdots\mathbin\Vert
 y^{(\lambda)}_{\pi_\lambda(t),e_{\lambda,t}},\\
 Y_{t,1}&=y^{(1)}_{\pi_1(t),1}\mathbin\Vert
 y^{(2)}_{\pi_2(t),1\oplus e_{2,t}}\mathbin\Vert\cdots\mathbin\Vert
 y^{(\lambda)}_{\pi_\lambda(t),1\oplus e_{\lambda,t}}.
\end{align*}
The important observation is that guessing $Y_{t,0},Y_{t,1}$ implies also guessing \emph{all} $\lambda$ constituents simultaneously. We can therefore hope to conclude that, if each individual claw is $\frac14$ secure, then the probability of guessing the combined claw state must be bounded by $4^{-\lambda}+\negl(\lambda)$.

\paragraph{Idea III: Sequential Repetition for Quantum Search Games.}
To make this intuition formal, we prove a general \emph{sequential repetition theorem for quantum search games}: A quantum search game is an interactive protocol between a prover and a verifier, where the goal of the prover is to guess some secret $x$, and the verifier can either accept or reject. In $\ell$ sequential copies, the verifier uses independent fresh randomness for each, the adversary retains arbitrary quantum memory, and it may postpone \emph{all} predictions until every interaction has finished. We show that the combined success probability degrades exponentially with $\ell$. As usual in this context, the difficulty is that conditioning on success in other copies can disturb the quantum memory, and the verifier's secret target is unavailable to a reduction.

For any round $j$, we denote by $q_j$ the maximum probability, over arbitrary starting states just before copy $j$, of guessing all targets from $j$ through $\ell$ using the remaining strategy. Our goal will be to prove that $q_j\lesssim\varepsilon q_{j+1}$ where $\varepsilon$ is the bound on the single search game, with a reduction to the bound on the single-copy search game. We can imagine giving the reduction the best-possible quantum state as non-uniform advice, then let the reduction interact with the challenger for the $j$-th copy and simulate the rest of the interactions locally. If $q_{j+1}$ is sufficiently close to one, simply running that suffix already contradicts the one-copy bound. Otherwise, we need to amplify the successful part of the suffix \emph{coherently}. We do this using the quantum singular value transformation (QSVT)~\cite{GSLW19}, which is a powerful \emph{quantum rewinding} technique. In particular, rewinding enables us to effectively condition on the adversary successfully answering targets $j+1$ through $\ell$ (note that the reduction knows what these targets are, the only missing information is target $j$), and only then output its guess for target $j$. This increases its overall probability of success by a factor of roughly $1/q_{j+1}$, so if we had assumed for contradiction that $q_j > \varepsilon q_{j+1}$, then our reduction outputs the $j$-th target with probability $> \varepsilon$, a contradiction. We refer the reader to the technical sections for more details.

\paragraph{Putting Things Together.}
Taking a step back, we have therefore obtained, from quantum-secure one-way functions, a $\poly(\lambda)$-round claw-state generation protocol for any polynomial $T$: At the end of the interaction, the honest server holds $\bigotimes_t\frac{\ket{0,Y_{t,0}}+\ket{1,Y_{t,1}}}{\sqrt2}$, and outputting any classical pair is computationally hard. Quantum communication is only client-to-server, with $\poly(\lambda,\log T)$ preparation cost and circuits fixed before interaction. The \cite{BK25} transformation converts these correlations into a blind delegation protocol with the same efficiency properties, which supplies the missing ingredient for the communication-compression compiler of \cite{BLM26}. Combining all the steps, we obtain a succinct argument for QMA from collapsing hash functions. 

\subsection{Concurrent work}

In a concurrent and independent work, \cite{cryptoeprint:2026/2099} also construct succinct arguments for QMA in Minicrypt via a completely different approach. In particular, they compile the quantum IOP of \cite{sun_et_al:LIPIcs.CCC.2026.4} into a succinct argument for QMA that is proven unconditionally secure in the quantum random oracle model. They do not claim a construction in the plain model.

\subsection{AI usage}
The overall architecture and general proof strategy were developed without AI assistance. ChatGPT was used to refine the protocol, assist with technical proofs, revise the exposition, and check references. The authors take full responsibility for the content.

\section{Preliminaries}
Let $\lambda$ be the security parameter.
We write $[n] = \{1, \dots, n\}$ and use logarithms to base $2$.
For a finite set $S$, we write $s\sim S$ when $s$ is sampled uniformly from $S$.
All inner products between bit strings are over $\mathbb F_2$.

All registers are finite-dimensional Hilbert spaces.
We consider security against non-uniform QPT adversaries that may initialize a polynomial-size register to an arbitrary non-uniform state independent of fresh honest randomness. All cryptographic assumptions and computational indistinguishability statements quantify over adversaries with such auxiliary input; negligible terms may depend on the adversary. Instance and circuit sizes are polynomially bounded in the security parameter, while the polynomials in succinctness bounds are fixed independently of the original verification time and witness size.

We use sans-serif letters such as $\mathsf{A},\mathsf{B},\mathsf{R}$ for registers, $\mathcal H_{\mathsf{A}}$ for the Hilbert space of register $\mathsf{A}$, and $I_{\mathsf{A}}$ for its identity operator. Register labels on states and operators indicate where they act; omitted registers are left unchanged. Braces with subscripts denote indexed lists. We write $\Vert$ for concatenation of bit strings, $\|\cdot\|$ for vector or operator norm, and
\[
 \|A\|_1:=\operatorname{Tr}\sqrt{A^\dagger A}
\]
for the trace norm. The trace distance of states $\rho,\sigma$ is $\tfrac12\|\rho-\sigma\|_1$. A POVM is a family of positive operators summing to the identity. We write $\approx_c$ for computational indistinguishability and $\negl(\lambda)$ for a function smaller than every inverse polynomial.

We use compressing collapsing hash functions in the sense of~\cite{Unr16}: After a hash is evaluated coherently and its output is measured, additionally measuring the input is computationally undetectable, even given the remaining workspace. We use this property through the compiler of~\cite{BLM26}. All one-way functions and pseudorandom functions are assumed secure against the quantum adversaries specified above.

\subsection{Succinct Arguments for QMA}

Fix a QMA promise problem $L=(L_{\mathrm{yes}},L_{\mathrm{no}})$, with a polynomial-time uniform quantum verifier $V=\{V_x\}$. For $x\in L_{\mathrm{yes}}$, there exists a polynomial-size quantum witness $\rho$ such that $V_x(\rho)$ accepts with probability at least $2/3$; for $x\in L_{\mathrm{no}}$, no witness causes $V_x$ to accept with probability greater than $1/3$. We recall the definition of succinct arguments for QMA below~\cite{BKL+22}.

\begin{definition}[Succinct Argument for QMA]
A succinct argument for the QMA promise problem $L$ is an interactive protocol $\Pi$ between a prover and a verifier where the common input is $(1^\lambda, x)$ and the prover additionally receives a quantum state $\rho$. The protocol satisfies the following properties.
\begin{itemize}
    \item Correctness: For $x \in L_{\mathrm{yes}}$, there exists a QPT prover that, on input a valid witness $\rho$, causes the verifier to accept with probability at least $\frac{2}{3}$.
    \item Soundness: For $x \in L_{\mathrm{no}}$ and all QPT provers on input an arbitrary state $\rho$, the probability that the verifier accepts is bounded by $\frac{1}{3}$.
    \item Succinctness: The communication complexity of the protocol is bounded by $\poly(\lambda)$ and the verifier's complexity is bounded by $|x|\cdot\poly(\lambda)$.
\end{itemize}
\end{definition}

\subsection{Claw States}
We refer to states of the form
\[
\ket{\psi_{x_0,x_1}}:=\frac{\ket{0,x_0}+\ket{1,x_1}}{\sqrt2}
\]
for $x_0,x_1\in\{0,1\}^n$, as \emph{claw states}. Let $\rho_{x_0,x_1} := \ketbra{\psi_{x_0,x_1}}{\psi_{x_0,x_1}}$, it is easy to see that $\rho_{x_0,x_1}$ can be prepared by a polynomial-size quantum circuit, given $(x_0,x_1)$.
The following lemma shows that a uniformly sampled claw state is hard to guess.
\begin{lemma}\label{lem:oneway}
    Consider the following experiment:
    \begin{itemize}
        \item Sample $x_0,x_1\sim\{0,1\}^n$, then prepare the state $\rho_{x_0, x_1}$ and send it to the adversary.
        \item The adversary returns $(x_0^*, x_1^*)$.
        \item The adversary wins if $(x_0, x_1) = (x_0^*, x_1^*)$.
    \end{itemize}
    The probability that any adversary succeeds in this experiment is at most $2^{1-n}$.
\end{lemma}
\begin{proof}
Let $\mathsf{C}$ be the claw register, with $\mathcal H_{\mathsf{C}}\cong\mathbb C^{2^{n+1}}$, and let $\mathsf{A}$ be the adversary's auxiliary register. We may assume without loss of generality that its initial state is pure, say $\ket{\alpha}_{\mathsf{A}}$: a mixed state can be purified by enlarging $\mathsf{A}$. This state is independent of $(x_0,x_1)$.

Let $\{M_{x_0,x_1}\}_{x_0,x_1}$ be the POVM on $\mathsf{C}\mathsf{A}$ describing the adversary's guess. To obtain a measurement on the claw register alone, define
\[
 N_{x_0,x_1}
 :=(I_{\mathsf{C}}\otimes\bra{\alpha}_{\mathsf{A}})
 M_{x_0,x_1}
 (I_{\mathsf{C}}\otimes\ket{\alpha}_{\mathsf{A}}).
\]
Each $N_{x_0,x_1}$ is positive, and
$
 \sum_{x_0,x_1}N_{x_0,x_1} =I_{\mathsf{C}},
$
thus these operators form a POVM. They give exactly the same guessing probabilities because
\[
 \operatorname{Tr}\!\left(M_{x_0,x_1}
    (\rho_{x_0,x_1}\otimes\ketbra{\alpha}{\alpha})\right)
 =\operatorname{Tr}(N_{x_0,x_1}\rho_{x_0,x_1}).
\]
The adversary's success probability is therefore
\begin{align*}
 \frac1{2^{2n}}\sum_{x_0,x_1}
 \operatorname{Tr}(N_{x_0,x_1}\rho_{x_0,x_1})
 &\leq\frac1{2^{2n}}\sum_{x_0,x_1}\operatorname{Tr}(N_{x_0,x_1})\\
 &=\frac1{2^{2n}}\operatorname{Tr}\!\left(\sum_{x_0,x_1}N_{x_0,x_1}\right)\\
 &=\frac1{2^{2n}}\operatorname{Tr}(I_{\mathsf{C}})\\
 &=2^{1-n},
\end{align*}
where the inequality uses $\rho_{x_0,x_1}\preceq I_{\mathsf{C}}$.
\end{proof}

\subsection{Pseudorandom Functions}

Let $T$ be a power of $2$ and $f_k :\{0,1\}^{\log T} \to \{0,1\}^{2n}$ be a keyed function, for $k \in\{0,1\}^\lambda$. We say that $f_k$ is a pseudorandom function if its table of outputs is computationally indistinguishable from a uniformly random table. In this work we consider the special case where $T=\poly(\lambda)$ and we assume that the function satisfies the following notion of pseudorandomness, which is implied by the regular definition of a pseudorandom function.
\begin{definition}[Pseudorandomness]\label{def:prf}
    A function $f_k:\{0,1\}^{\log T} \to \{0,1\}^{2n}$ is pseudorandom if, for all $x^* \in \{0,1\}^{\log T}$, the following distributions are computationally indistinguishable
\[
\bigl(\{f_k(x)\}_{x\ne x^*},f_k(x^*)\bigr)
 \approx_c
 \bigl(\{f_k(x)\}_{x\ne x^*},y^*\bigr).
\]
    Here $k\sim\{0,1\}^\lambda$ and $y^*\sim\{0,1\}^{2n}$ are independent, and the non-target entries are listed in the same fixed order on both sides.
\end{definition}
Such pseudorandom functions can be constructed assuming the existence of (quantum-secure) one-way functions \cite{Zhandry12}.

\section{Claw-State Generation Protocol}

\subsection{Protocol Description}\label{sec:protocol}
Set $n = \lambda$ and let $T=\poly(\lambda)$ be a power of $2$. We consider the following interactive protocol between a client and a server.
\begin{itemize}
    \item[(i)] The client samples $k \sim \{0,1\}^\lambda$, then for all $i\in[\log T]$ and $b\in\{0,1\}$ it samples $x_{i,b}\sim\{0,1\}^n$ and $k_{i,b}\sim\{0,1\}^\lambda$.
    \item[(ii)] For all $t\in\{0,1\}^{\log T}$ and $i\in[\log T]$, the client sends to the server
\[
c_{t, i}:= f_k(t) \oplus f_{k_{i,1-t_i}}(t) \in\{0,1\}^{2n}.
\]
\item[(iii)] The client computes the state  
\[
\bigotimes_{i\in[\log T]}\left(\frac{\ket{0,x_{i,0}, k_{i,0}} + \ket{1,x_{i,1}, k_{i,1}}}{\sqrt{2}}\right) = \frac{1}{\sqrt{T}} \sum_{t \in\{0,1\}^{\log T}} \bigotimes_{i\in[\log T]}\ket{t_i,x_{i,t_i},k_{i,t_i}}.
\]
Let $(y_{t,0}, y_{t,1}) = f_k(t)$. The client applies the isometry
\[
\ket{t}\ket0 \mapsto \ket{t}\ket{\psi_{y_{t,0}, y_{t,1}}}
\]
to the leading qubits of the above state and a fresh output register to obtain
\[
\frac{1}{\sqrt{T}} \sum_{t \in\{0,1\}^{\log T}} \left(\bigotimes_{i\in[\log T]}\ket{t_i,x_{i,t_i},k_{i,t_i}}\right)\otimes \ket{\psi_{y_{t,0}, y_{t,1}}}.
\]
The client returns its auxiliary qubits to zero and sends the resulting state to the server.
\item[(iv)] The server maps the state received from the client to
\[
\left(\frac{1}{\sqrt{T}} \sum_{t \in\{0,1\}^{\log T}} \bigotimes_{i\in[\log T]}\ket{t_i,x_{i,t_i},k_{i,t_i}}\right) \bigotimes_{t \in\{0,1\}^{\log T}}\ket{\psi_{y_{t,0}, y_{t,1}}}.
\]
Note that this isometry can be implemented efficiently, given $\{c_{t,i}\}_{t, i}$. Indeed, for any fixed $t$ and any $t'\neq t$, the server can recover $f_{k}(t')$ by selecting an index $i$ such that $t_i' \neq t_i$, so that $k_{i, 1-t_i'} =k_{i, t_i}$, and computing $c_{t', i}\oplus f_{k_{i,t_i}}(t')= f_k(t')  = (y_{t',0}, y_{t',1})$. Then the server can compute the claw state $\ket*{\psi_{y_{t',0}, y_{t',1}}}$ by evaluating the claw-preparation circuit. Running this algorithm coherently and reordering the registers leads to the state as described above.
\item[(v)] The server measures its first registers in the Hadamard basis to obtain $d=(d_1, \dots, d_{\log T})$, where $d_i\in\{0,1\}\times \{0,1\}^{n}\times \{0,1\}^{\lambda}$, which are sent to the client.
\item[(vi)] The client parses $d_i = (d_i^{(0)}, d_i^{(1)}, d_i^{(2)})$, and for all $i\in[\log T]$ it checks that $d_i^{(1)} \neq 0$ and that
\[
\langle d_i, (1, x_{i,0}\oplus x_{i,1}, k_{i,0}\oplus k_{i,1})\rangle =0.
\]
We say that the server \emph{passes} the protocol if all checks of the client succeed. If this is the case, the client returns $\{y_{t,0},y_{t,1}\}_{t\in\{0,1\}^{\log T}}$ as its private local output.
\end{itemize}
Note that the client's quantum operations are confined in step (iii) and their runtime is bounded by some $\mathrm{poly}(\log T, n, \lambda)$. On the other hand, the client's classical runtime as well as the server's total runtime is bounded by $\mathrm{poly}(T, n, \lambda)$. 

Moreover, it is easy to see that an honest server passes the protocol except with probability at most $(\log T)2^{-n}$, since the probability that $d_i^{(1)} = 0$ is exactly $2^{-n}$.

\subsection{Analysis}

Next, we show that the output claw states are hard to guess. Assuming $n=\Omega(\lambda)$, we show the following.
\begin{theorem}\label{thm:single}
Consider the following experiment:
\begin{itemize}
\item Run the client-server protocol from \Cref{sec:protocol}.
\item After the server sends $d$, the client samples $t^*\sim\{0,1\}^{\log T}$ and sends this index and the strings
$\{y_{t,0},y_{t,1}\}_{t\ne t^*}$
to the server.
\item The server replies with $(y_0^*,y_1^*)$ and wins if
$(y_0^*,y_1^*)=(y_{t^*,0},y_{t^*,1})$
and $d$ passes the client's checks.
\end{itemize}
Assuming quantum-secure one-way functions, for every QPT server there is a negligible function $\mu$ such that its success probability is at most $1/4+\mu(\lambda)$.
\end{theorem}
\begin{proof}
Let $X$ contain the client's random choices of $k$ and $\{x_{i,b},k_{i,b}\}_{i\in[\log T],\,b\in\{0,1\}}$, and let $\mathrm{Valid}_X(d)$ be its verification predicate. We use the following registers throughout this proof and the two lemmas below: $\mathsf{C}$ contains the client's quantum message, $\mathsf{A}$ contains the server's pure auxiliary input $\ket{\alpha}$, $\mathsf{D}$ contains the client's copy of $d$, $\mathsf{B}$ contains all registers retained by the server after sending $d$, and $\mathsf{L}$ contains the client's final classical message.

Let $\mathcal A:\mathsf{C}\mathsf{A}\to\mathsf{D}\mathsf{B}$ be a unitary dilation of the server's first stage. Its dependence on the classical table $\{c_{t,i}\}_{t,i}$ is implicit. The server may retain its own copy of $d$ in $\mathsf{B}$, but all its subsequent operations act trivially on $\mathsf{D}$. Define
\[
\Pi_X:=\sum_{d:\mathrm{Valid}_X(d)=1}\ketbra{d}{d}_{\mathsf{D}}\otimes I_{\mathsf{B}}.
\]
For fixed $X$, the client sends its final message by applying the isometry
\[
V_t\ket{\varphi}_{\mathsf{DB}}
=
\ket{\varphi}_{\mathsf{DB}}\otimes
\ket{t,\{y_{u,0},y_{u,1}\}_{u\ne t}}_{\mathsf{L}}.
\]
The dependence of $V_t$ on $X$ is implicit. Let $\{M_{t,y_0,y_1}\}_{y_0,y_1\in\{0,1\}^n}$ be the POVM on $\mathsf{B}\mathsf{L}$ describing the server's second message, extended by the identity on $\mathsf{D}$.

We omit register subscripts in the remaining state formulas. The subnormalized state for which the client accepts, right before the final measurement, is
\[
\ket{\xi_{X,t^*}}
:=\frac1{\sqrt T}V_{t^*}\sum_t\Pi_X\mathcal A
\left(
\left(\bigotimes_{i\in[\log T]}\ket{t_i,x_{i,t_i},k_{i,t_i}}\right)
\ket{\psi_{y_{t,0},y_{t,1}}}\ket{\alpha}
\right).
\]
Define also the possibly subnormalized state
\[
\ket{\zeta_{X,t^*}}
:=V_{t^*}\Pi_X\mathcal A
\left(
\left(\bigotimes_{i\in[\log T]}\ket{t_i^*,x_{i,t_i^*},k_{i,t_i^*}}\right)
\ket{\psi_{y_{t^*,0},y_{t^*,1}}}\ket{\alpha}
\right).
\]
Both states are on $\mathsf{D}\mathsf{B}\mathsf{L}$. We use the decomposition
\[
\ket{\xi_{X,t^*}}
=\left(\ket{\xi_{X,t^*}}-\frac{\sqrt T}{2}\ket{\zeta_{X,t^*}}\right)
+\frac{\sqrt T}{2}\ket{\zeta_{X,t^*}}.
\]
Intuitively, the second summand isolates the contribution obtained when the server is given the single claw state $\ket*{\psi_{y_{t^*,0},y_{t^*,1}}}$, which we will bound by the hardness of recovering both strings from one claw state. The residual term can be bounded using only the verification constraints. The coefficient $\sqrt T/2$ is chosen so that, after averaging over $t^*$, the cross term in the expansion of the residual squared norm cancels its $\|\ket{\xi_{X,t^*}}\|^2$ term.

By Cauchy--Schwarz and using $0\preceq M_{t^*,y_{t^*,0},y_{t^*,1}}\preceq I$, the square root of the prover's total success probability satisfies
\begin{align}\label{eq:single-success-bound}
&\sqrt{\frac1T\sum_{t^*}\E_X
\bra{\xi_{X,t^*}}M_{t^*,y_{t^*,0},y_{t^*,1}}\ket{\xi_{X,t^*}}}\\
&\quad\leq
\sqrt{\frac1T\sum_{t^*}\E_X
\left\|\ket{\xi_{X,t^*}}-\frac{\sqrt T}{2}\ket{\zeta_{X,t^*}}\right\|^2}+\frac12\sqrt{\sum_{t^*}\E_X
\bra{\zeta_{X,t^*}}M_{t^*,y_{t^*,0},y_{t^*,1}}\ket{\zeta_{X,t^*}}}.\notag
\end{align}
For fixed $X$, the definitions and the isometry property give
\[
V_t^\dagger\ket{\xi_{X,t}}
=\frac1{\sqrt T}\sum_u V_u^\dagger\ket{\zeta_{X,u}}.
\]
In particular, the left-hand side is independent of $t$. Thus
\[
\sum_t\operatorname{Re}\braket{\xi_{X,t}}{\zeta_{X,t}}
=\sqrt T\,\|\ket{\xi_{X,t^*}}\|^2
\qquad\text{for every }t^*.
\]
Expanding the squared norm on the right-hand side of \Cref{eq:single-success-bound} now gives
\begin{align*}
\frac1T\sum_t\left\|\ket{\xi_{X,t}}-\frac{\sqrt T}{2}\ket{\zeta_{X,t}}\right\|^2 &=\frac1T\sum_t\|\ket{\xi_{X,t}}\|^2
+\frac14\sum_t\|\ket{\zeta_{X,t}}\|^2
-\frac1{\sqrt T}\sum_t\operatorname{Re}\braket{\xi_{X,t}}{\zeta_{X,t}}\\
&=\frac14\sum_t\|\ket{\zeta_{X,t}}\|^2.
\end{align*}
Averaging over $X$ and applying \Cref{lem:preceding} (proven below) bounds the expression under the first square root on the right-hand side of \Cref{eq:single-success-bound} by $1/4$. For the second term, \Cref{lem:preceding2} (proven below) gives
\[
\frac14\sum_{t^*}\E_X
\bra{\zeta_{X,t^*}}M_{t^*,y_{t^*,0},y_{t^*,1}}\ket{\zeta_{X,t^*}}
\le\frac{2^{1-n}+\mu(\lambda)}4.
\]
Combining these bounds, the success probability is at most
\[
\left(\frac12+\frac12\sqrt{2^{1-n}+\mu(\lambda)}\right)^2
=\frac14+\negl(\lambda),
\]
since $n=\Omega(\lambda)$.
\end{proof}

Using the notation introduced above, we prove the two remaining technical statements in the following.
\begin{lemma}\label{lem:preceding}
\[
\sum_{t\in\{0,1\}^{\log T}}\E_X\|\ket{\zeta_{X,t}}\|^2\le1.
\]
\end{lemma}
\begin{proof}
Fix $t\in\{0,1\}^{\log T}$ and condition on all client randomness except $\{x_{i,1-t_i}\}_{i\in[\log T]}$. Write
\[
\mathcal A\left(
\left(\bigotimes_{i\in[\log T]}\ket{t_i,x_{i,t_i},k_{i,t_i}}\right)
\ket{\psi_{y_{t,0},y_{t,1}}}\ket{\alpha}
\right)
=\sum_d\ket d\ket{\phi_{t,d}}
\]
for possibly subnormalized states $\ket{\phi_{t,d}}$ on $\mathsf{B}$, with $\ket d$ on $\mathsf{D}$. These states are fixed under the conditioning: the input state contains only $\{x_{i,t_i}\}_{i\in[\log T]}$, and the classical table $\{c_{t,i}\}_{t,i}$ depends only on the keys. Since $V_t$ is an isometry,
\begin{align*}
\E_{\{x_{i,1-t_i}\}_i}\|\ket{\zeta_{X,t}}\|^2
&=\E_{\{x_{i,1-t_i}\}_i}
\left\|\Pi_X\sum_d\ket d\ket{\phi_{t,d}}\right\|^2\\
&=\sum_d\E_{\{x_{i,1-t_i}\}_i}\mathbf1_{\mathrm{Valid}_X(d)=1}\,
\|\ket{\phi_{t,d}}\|^2\\
&\le\frac1T\sum_d\|\ket{\phi_{t,d}}\|^2
=\frac1T.
\end{align*}
The second equality uses the orthogonality of the states $\ket d$, and the inequality follows from
\begin{equation}\label{eq:hidden-average}
\E_{\{x_{i,1-t_i}\}_i}\mathbf1_{\mathrm{Valid}_X(d)=1}
=\frac1T\mathbf1_{\forall i\in[\log T],\,d_i^{(1)}\ne0}.
\end{equation}
Indeed, if $d_i^{(1)}=0$ for some $i$, both sides vanish. Otherwise, each verification check is a nonconstant affine equation in the independent uniform string $x_{i,1-t_i}$, and hence holds with probability exactly $1/2$. There are $\log T$ independent checks, so their joint probability is $2^{-\log T}=1/T$.

Averaging over the remaining client randomness and summing over $t$ proves the claim.
\end{proof}

\begin{lemma}\label{lem:preceding2}
Assuming quantum-secure one-way functions, for every QPT server there is a negligible function $\mu$ such that
\[
\sum_{t\in\{0,1\}^{\log T}}\E_X
\bra{\zeta_{X,t}}M_{t,y_{t,0},y_{t,1}}\ket{\zeta_{X,t}}
\le2^{1-n}+\mu(\lambda).
\]
\end{lemma}
\begin{proof}
Consider the following modified experiment:
\begin{itemize}
\item Sample $t^*\sim\{0,1\}^{\log T}$, then generate the client's complete setup and classical table $\{c_{t,i}\}_{t,i}$ as in \Cref{sec:protocol}, except that the client sends the state
\[
\left(\bigotimes_{i=1}^{\log T}
\ket{t_i^*,x_{i,t_i^*},k_{i,t_i^*}}\right)
\otimes\ket{\psi_{y_{t^*,0},y_{t^*,1}}}
\]
to the server.
\item After the server sends $d$, the client always sends $t^*$ and $\{y_{t,0},y_{t,1}\}_{t\ne t^*}$ to the server.
\item The server replies with $(y_0^*,y_1^*)$. It wins if $d_i^{(1)}\ne0$ for all $i\in[\log T]$ and
\[
(y_0^*,y_1^*)=(y_{t^*,0},y_{t^*,1}).
\]
\end{itemize}
We first show that the success probability in this modified experiment equals the sum in the statement. Fix $t$ and condition on all client randomness except $\{x_{i,1-t_i}\}_{i\in[\log T]}$. As in \Cref{lem:preceding}, write
\[
\mathcal A\left(
\left(\bigotimes_{i\in[\log T]}\ket{t_i,x_{i,t_i},k_{i,t_i}}\right)
\ket{\psi_{y_{t,0},y_{t,1}}}\ket{\alpha}
\right)
=
\sum_d\ket d\ket{\phi_{t,d}}.
\]
The states $\ket{\phi_{t,d}}$, the isometry $V_t$, and the operator $M_{t,y_{t,0},y_{t,1}}$ are fixed under this conditioning. Since $V_t$ and $M_{t,y_{t,0},y_{t,1}}$ act trivially on $\mathsf{D}$ and $\braket{d'}{d}=0$ for $d'\ne d$, all terms with different values of $d$ vanish. Therefore,
\begin{align*}
\bra{\zeta_{X,t}}M_{t,y_{t,0},y_{t,1}}\ket{\zeta_{X,t}}=
\sum_d\mathbf1_{\mathrm{Valid}_X(d)=1}\cdot
\bigl(\bra d\bra{\phi_{t,d}}\bigr)
V_t^\dagger M_{t,y_{t,0},y_{t,1}} V_t
\bigl(\ket d\ket{\phi_{t,d}}\bigr).
\end{align*}
Every quadratic form on the right is nonnegative and independent of the hidden strings. Applying \Cref{eq:hidden-average} therefore gives
\begin{align*}
\E_{\{x_{i,1-t_i}\}_i}
\bra{\zeta_{X,t}}M_{t,y_{t,0},y_{t,1}}\ket{\zeta_{X,t}}=
\frac1T\sum_{d:\forall i,\,d_i^{(1)}\ne0}
\bigl(\bra d\bra{\phi_{t,d}}\bigr)
V_t^\dagger M_{t,y_{t,0},y_{t,1}} V_t
\bigl(\ket d\ket{\phi_{t,d}}\bigr).
\end{align*}
Averaging over the remaining randomness and summing over $t$ yields
\begin{align*}
\sum_t\E_X\bra{\zeta_{X,t}}M_{t,y_{t,0},y_{t,1}}\ket{\zeta_{X,t}}=
\frac1T\sum_t\E_X
\sum_{d:\forall i,\,d_i^{(1)}\ne0}
\bigl(\bra d\bra{\phi_{t,d}}\bigr)
V_t^\dagger M_{t,y_{t,0},y_{t,1}} V_t
\bigl(\ket d\ket{\phi_{t,d}}\bigr).
\end{align*}
The right-hand side is exactly the success probability of the modified experiment: the factor $1/T$ samples $t^*$ uniformly, the restriction on $d$ enforces the nonzero checks, and the POVM element enforces the correct output.

It remains to bound this probability. Fix a target $t^*\in\{0,1\}^{\log T}$ and consider the following hybrids:
\begin{itemize}
\item \textbf{Hybrid 0:} This is the modified experiment above, with target $t^*$.
\item \textbf{Hybrid 1:} Replace the entries $\{c_{t^*,i}\}_{i\in[\log T]}$ by independent uniform strings.

We make these replacements one at a time. For the $i$th replacement, the masking key $k_{i,1-t_i^*}$ does not occur in the state sent to the server. A reduction uses the PRF challenge table for this key to construct the entire table, changing only its value at $t^*$. It samples the other keys itself and simulates all remaining messages. Because the parity checks have been omitted, the simulation never needs the key $k_{i,1-t_i^*}$ itself. Thus, \Cref{def:prf} makes each replacement computationally indistinguishable.
\item \textbf{Hybrid 2:} Replace $(y_{t^*,0},y_{t^*,1})=f_k(t^*)$ by an independent uniform pair.

The $t^*$-th row of the classical table is now uniform and independent of this pair. A reduction can therefore use its PRF challenge table for the master key $k$ to prepare the target claw and all remaining messages. Again, only the value at $t^*$ changes, so \Cref{def:prf} makes this hybrid computationally indistinguishable from Hybrid 1.
\end{itemize}
Note that these simulations are efficient because the PRF domain has size $T=\poly(\lambda)$.

In Hybrid 2, the target pair is uniform, and the server receives only its claw state together with data independent of that pair. Even after dropping the nonzero checks, its success probability is at most $2^{1-n}$, by \Cref{lem:oneway}. The $\log T+1$ replacements change this probability by a negligible amount. Averaging over $t^*$ proves the statement; non-uniform security makes the negligible bound uniform over the polynomially many targets.
\end{proof}

\section{Sequential Repetition}

In this section, we give a generic sequential repetition theorem for any quantum interactive search game.

\begin{definition}[Quantum interactive search game]\label{def:qsearch}
A quantum interactive search game $\Gamma$ is specified by a polynomial-size quantum
interactive circuit family $V=\{V_\secp\}_\secp$. At the end of an execution it outputs a classical string
\[
  x\in \{0,1\}^{m(\secp)}\cup\{\bot\}
\]
in a designated classical register $\mathsf{X}$. 

For any polynomial-size interactive strategy $B$ with input register $\mathsf{M}$
and output register $\mathsf{M}'$, interaction with 
$V_\secp$ defines the map
\begin{equation}\label{eq:exec-map}
  \Exec_{\Gamma,\secp}[B] :
       \mathsf{M} \longrightarrow \mathsf{M}'\mathsf{X}.
\end{equation}
\end{definition}

\begin{definition}[Success probability]
    A QPT adversary $\cB = (\ket{\psi},B,B_\out)$ against a quantum interactive search game $\Gamma$ consists of a polynomial-size advice family $\{\ket{\psi_\secp}\}_{\secp}$, a polynomial-size interactive strategy $B$, and a polynomial-size predictor circuit $B_\out$ that takes register $\mathsf{M}'$ as input and produces a classical string $y \in \{0,1\}^{m(\secp)}$ in a designated classical register $\mathsf{Y}$. The success probability $p_{\Gamma,\secp}(\cB)$ of the adversary is defined as 

\[p_{\Gamma,\secp}(\cB) \coloneqq \operatorname{Tr}\!\left[\Pi_\win (B_\out \otimes I_{\mathsf{X}})\Exec_{\Gamma,\secp}[B](\ketbra{\psi_\secp})\right],\]

where \[\Pi_\win \coloneqq \sum_{x \in \{0,1\}^{m(\secp)}}\ketbra{x}_{\mathsf{X}} \otimes \ketbra{x}_{\mathsf{Y}}.\]
The winning projector acts as the identity on unmentioned registers.
\end{definition}

\begin{definition}[Sequential success probability]
    
Given a repetition parameter $\ell = \ell(\secp)$, a QPT adversary $\cA = (\ket{\psi},A_1,\dots,A_\ell,A_\out)$ against the $\ell$-sequentially-repeated quantum interactive search game $\Gamma$ has success probability $p_{\Gamma,\secp}^\ell(\cA)$ defined as follows. Each copy uses fresh verifier randomness and workspace.
\begin{itemize}
    \item Set $\sigma_0 \coloneqq \ketbra{\psi_\secp}$ on register $\mathsf{M}_0$.
    \item For $i \in [\ell]$, set $\sigma_i \coloneqq \left(\Exec_{\Gamma,\secp}[A_i]_{\mathsf{M}_{i-1} \to \mathsf{M}_i\mathsf{X}_i} \otimes I_{\mathsf{X}_1,\dots,\mathsf{X}_{i-1}}\right)(\sigma_{i-1})$.
    \item $p_{\Gamma,\secp}^\ell(\cA) \coloneqq \operatorname{Tr}\!\left[\Pi_\win^\ell(A_\out \otimes I_{\mathsf{X}_1,\dots,\mathsf{X}_\ell})(\sigma_\ell)\right]$,
\end{itemize}
 where \[\Pi_\win^\ell \coloneqq \sum_{x_1,\dots,x_\ell\in\{0,1\}^{m(\secp)}}\ketbra{x_1,\dots,x_\ell}_{\mathsf{X}_1,\dots,\mathsf{X}_\ell} \otimes \ketbra{x_1,\dots,x_\ell}_{\mathsf{Y}_1,\dots,\mathsf{Y}_\ell}.\]

\end{definition}

\begin{theorem}\label{thm:sequential-repetition}
    For any quantum interactive search game $\Gamma$, any polynomial $\ell = \ell(\secp)\geq1$, and any $\varepsilon(\secp)\in[0,1]$, if for all QPT adversaries $\cB$, $p_{\Gamma,\secp}(\cB) \leq \varepsilon(\secp) + \negl(\secp)$, then for all QPT adversaries $\cA$ against the $\ell$-sequentially-repeated game, $p_{\Gamma,\secp}^\ell(\cA) \leq \varepsilon(\secp)^\ell + \negl(\secp)$.
\end{theorem}

The main quantum information tool that we use to prove \cref{thm:sequential-repetition} is the following uniform singular value amplification result from \cite{GSLW19}.

\begin{lemma}[{\cite[Theorem~30 and the paragraph immediately following it]{GSLW19}}]\label{lemma:GSLW}
    Let $U$ be a unitary and $\Pi_\inp,\Pi_\out$ be two projectors such that $U,U^\dagger$ and coherent tests of the projectors are implementable by quantum circuits of size at most $s$. Let $\gamma > 1, \delta,\zeta \in (0,1/2)$, define $K \coloneqq \Pi_\out U \Pi_\inp$, and suppose that $\|K\| \leq \frac{1-\delta}{\gamma}$. Then there exists a unitary $W$, using one additional qubit $\mathsf{C}$ and implementable by a circuit of size at most $\poly(s,\gamma,1/\delta,\log(1/\zeta))$, such that \[\|\left(\bra{0}_{\mathsf{C}} \otimes \Pi_\out\right)W\left(\ket{0}_{\mathsf{C}} \otimes \Pi_\inp\right) - \gamma K\| < \zeta.\]
\end{lemma}

Next we give a direct corollary stated in a manner that will be convenient for our proof of \cref{thm:sequential-repetition}.

\begin{corollary}\label{cor:GSLW}
    Let $U,\Pi_\inp,\Pi_\out,K,\delta,\zeta$ be as above and suppose that $c>0$ and $\|K\| \leq \sqrt{c} < 1-\delta$. Then there exists a unitary $W$ implementable by a circuit of size at most $\poly(s,1/c,1/\delta,\log(1/\zeta))$ such that the following holds. Define \[L \coloneqq \left(\bra{0}_{\mathsf{C}} \otimes \Pi_\out\right)W\left(\ket{0}_{\mathsf{C}} \otimes \Pi_\inp\right).\] For any register $\mathsf{R}$, normalized state $\ket{\psi}$ in the image of $\Pi_\inp \otimes I_{\mathsf{R}}$, and any projector $\Pi$, \[\|\Pi(L \otimes I_{\mathsf{R}})\ket{\psi}\|^2 \geq \frac{(1-\delta)^2}{c}\|\Pi(K \otimes I_{\mathsf{R}})\ket{\psi} \|^2 - \zeta.\]
\end{corollary}

\begin{proof}
    Set $\gamma' = \frac{1-\delta}{\sqrt{c}}, \zeta' = \zeta/2$, and apply \cref{lemma:GSLW} with parameters $\gamma',\delta,\zeta'$ to obtain $L$ as defined in the corollary statement such that $\|L - \gamma'K\| < \frac{\zeta}{2}$. Let \[\ket{\nu} \coloneqq \left((L-\gamma'K) \otimes I_{\mathsf{R}}\right)\ket{\psi}\] and note that $\|\ket{\nu}\| < \zeta/2$. 
    We have that 
    \begin{align*}
        \|\Pi(L \otimes I_{\mathsf{R}})\ket{\psi}\|^2 &= \|\gamma'\Pi(K \otimes I_{\mathsf{R}})\ket{\psi} + \Pi\ket{\nu} \|^2 \\ &= \|\gamma'\Pi(K \otimes I_{\mathsf{R}})\ket{\psi}\|^2 + 2\gamma'\operatorname{Re}\bra{\psi}(K^\dagger \otimes I_{\mathsf{R}})\Pi\ket{\nu} + \|\Pi\ket{\nu}\|^2 \\
        &\geq \|\gamma'\Pi(K \otimes I_{\mathsf{R}})\ket{\psi}\|^2 - 2\|\gamma'\Pi(K \otimes I_{\mathsf{R}})\ket{\psi}\|\|\Pi\ket{\nu}\| \\ &> \frac{(1-\delta)^2}{c}\|\Pi(K \otimes I_{\mathsf{R}})\ket{\psi}\|^2 - \zeta. 
    \end{align*}
    Here the last step uses $\|\gamma'K\|\leq1-\delta<1$ and $\|\Pi\ket{\nu}\|<\zeta/2$.
\end{proof}

Finally, we prove \cref{thm:sequential-repetition}.

\begin{proof} (of \cref{thm:sequential-repetition}) Fix a search game $\Gamma$, polynomial $\ell = \ell(\secp)$, $\varepsilon = \varepsilon(\secp)$ and any adversary $\cA = (\ket{\psi},A_1,\dots,A_\ell,A_\out)$. For $j \in [\ell]$, define $q_j$ as follows. Start immediately before copy $j$ and initialize $A_j$ with an arbitrary state. Run $A_j$ through $A_\ell$ interacting with $V$, run $A_\out$ to produce guesses $y_j,\dots,y_\ell$, and accept if $y_j = x_j,\dots,y_\ell = x_\ell$. Then $q_j$ is the maximum over all initial states of the probability of accepting. 

Define $q_{\ell+1} = 1$, and note that $p^\ell_{\Gamma,\secp}(\cA) \leq q_1 \leq q_2 \leq \dots \leq q_{\ell+1}$. The theorem follows if $\max_{j\in[\ell]}(q_j-\varepsilon q_{j+1})_+$ is negligible, where $(u)_+:=\max\{u,0\}$: iterating this uniform recurrence gives $q_1\leq\varepsilon^\ell+\ell\cdot\negl(\secp)$. Towards contradiction, suppose that there exist a polynomial $r(\secp)\geq1$, an infinite set of security parameters, and an index $j=j(\secp)\in[\ell]$ such that, on this infinite set, $q_j>\varepsilon q_{j+1}+\eta$, where $\eta\coloneqq1/r(\secp)$. Include $j$ in the non-uniform advice. Note that $q_{j+1} \geq q_j \geq \eta$, and consider the following two cases.

First suppose that $q_{j+1} \geq 1-\eta/16$. Define an adversary $\cB =(\ket{\phi},B,B_\out)$ against the $j$'th sequential game as follows. Let $\ket{\phi}$ be a pure state that attains $q_j$, $B = A_j$, and $B_\out$ run the remainder of $A_{j+1},\dots,A_\ell,A_\out$ while internally simulating $V$ in each game, and outputting the value $y_j$ output by $A_\out$. Then the probability that $y_j = x_j$ is at least

\[\varepsilon q_{j+1} + \eta \geq \varepsilon - \frac{\varepsilon \eta}{16} + \eta \geq \varepsilon + \frac{15\eta}{16},\] a contradiction.

Next suppose that $q_{j+1} < 1-\eta/16$. Let $U$ be a purification, retaining all measurement outcomes and discarded registers, of the procedure that runs $A_{j+1},\dots,A_\ell,A_\out$ while internally simulating $V$ in each game. Let $\mathsf{M}$ be the register passed to $A_{j+1}$ and $\mathsf{Z}$ be the auxiliary register needed to define the purification $U$, initialized to $\ket{0}_{\mathsf{Z}}$. Define $\Pi_\inp = I_{\mathsf{M}} \otimes \ketbra{0}_{\mathsf{Z}}$ and define $\Pi_\out$ to be the projector applied to the output of $U$ that checks that $y_i = x_i$ for all $i \in \{j+1,\dots,\ell\}$. Then defining \[K = \Pi_\out U \Pi_\inp,\] we have that $\|K\|^2 = q_{j+1}$. The circuits for $U,U^\dagger$ and coherent tests of both projectors have polynomial size.

Set $b=\lceil\log_2(32/\eta)\rceil$ and $c=2^{-b}\lceil2^bq_{j+1}\rceil$, included in the non-uniform advice, so that $q_{j+1}\leq c\leq q_{j+1}+\eta/32<1$, and apply \cref{cor:GSLW} with $\delta,\zeta = \eta/128$. Note that the hypothesis holds since $c < 1- \eta/32 < (1-\eta/128)^2 = (1-\delta)^2$. Since $c\geq\eta$, we obtain a polynomial-size circuit $W$ which we will use to define an adversary $\cB =(\ket{\phi},B,B_\out)$ as follows. Let $\ket{\phi}$ be a pure state that attains $q_j$ and $B = A_j$. Let $B_\out$ prepare $\ket{0}_{\mathsf{Z}}$ and an additional single-qubit register $\ket{0}_{\mathsf{C}}$, run $W$, and attempt to project the output onto $\ketbra{0}_{\mathsf{C}} \otimes \Pi_\out$. If successful it measures and outputs $y_j$. Otherwise, it outputs an arbitrary string. 

Let $\ket{\phi'}_{\mathsf{M},\mathsf{Z},\mathsf{X}_j,\mathsf{E}}$ be a purification of the output of $B$ interacting with $V$, tensored with $\ket{0}_{\mathsf{Z}}$, where $\mathsf{X}_j$ holds $V$'s output and $\mathsf{E}$ purifies the interaction. Then $B_\out$ acts on $\mathsf{M},\mathsf{Z},\mathsf{C}$ and we take $\mathsf{R} = \mathsf{X}_j\mathsf{E}$ to be the register that it does not touch. Let $\Pi$ be the projector checking that $x_j = y_j$ and let $L$ be as defined in the statement of \cref{cor:GSLW}. Then the probability that $\cB$ succeeds is at least \begin{align*}
\left\|\Pi(L\otimes I_{\mathsf{R}})\ket{\phi'}\right\|^2
&\geq \frac{(1-\delta)^2}{c}
       \left\|\Pi(K\otimes I_{\mathsf{R}})\ket{\phi'}\right\|^2-\zeta \\
&= (1-\zeta)^2\frac{q_j}{c}-\zeta
 \geq \frac{q_j}{c}-3\zeta \\
&\geq \frac{\varepsilon q_{j+1}+\eta}{c}-3\zeta
 = \varepsilon+\frac{\eta-\varepsilon(c-q_{j+1})}{c}-3\zeta \\
&\geq \varepsilon+\frac{\eta-\eta/32}{c}-3\zeta
 \geq \varepsilon+\frac{31\eta}{32}-3\zeta \\
&= \varepsilon+\frac{121\eta}{128},
\end{align*} a contradiction. 
\end{proof}

\section{Succinct Arguments for QMA}
Let $T = T(\lambda)$ be a polynomial in the security parameter and a power of $2$, and identify $[T]$ with $\{0,1\}^{\log T}$ using its canonical encoding. Consider the following protocol between a client and a (possibly malicious) server.
\begin{itemize}
    \item[(i)] Run $\lambda$ independent copies of the protocol in \Cref{sec:protocol} in sequence. Denote the output state held by the (honest) server in the $i$-th copy by
\[
\bigotimes_{t\in\{0,1\}^{\log T}}\ket{\psi_{y^{(i)}_{t,0},y^{(i)}_{t,1}}}
\]
up to an irrelevant global phase.
\item[(ii)] After the completion of all the protocols, the client samples independent uniform permutations $\pi_i \sim S_T$ and sends them to the server.
\item[(iii)] For every $t\in\{0,1\}^{\log T}$, the server takes the states indexed by $\pi_i(t)$. For each $i\in\{2,\ldots,\lambda\}$, it applies a CNOT from the first leading qubit to the $i$-th leading qubit, measures the latter in the computational basis, and records the outcome $e_{i,t}\in\{0,1\}$. After discarding the measured leading qubits, the remaining state is
$
\ket{\psi_{Y_{t,0},Y_{t,1}}},
$
where
\begin{equation}\label{eq:combined-endpoints}
\begin{aligned}
Y_{t,0}
&:=y^{(1)}_{\pi_1(t),0}\Vert
y^{(2)}_{\pi_2(t),e_{2,t}}\Vert\cdots\Vert
y^{(\lambda)}_{\pi_\lambda(t),e_{\lambda,t}},\\
Y_{t,1}
&:=y^{(1)}_{\pi_1(t),1}\Vert
y^{(2)}_{\pi_2(t),1\oplus e_{2,t}}\Vert\cdots\Vert
y^{(\lambda)}_{\pi_\lambda(t),1\oplus e_{\lambda,t}}.
\end{aligned}
\end{equation}
The server sends the outcomes $\{e_{i,t}\}_{i,t}$ to the client.

\item[(iv)] The client aborts if the verification in any of the copies fails and otherwise it locally computes $\{Y_{t,0},Y_{t,1}\}_{t}$ as specified above.
\end{itemize}
Clearly the only quantum communication between the client and the server is the one happening in the protocol from \Cref{sec:protocol}. Therefore, quantum states are only sent from the client to the server. Moreover, before any message is exchanged, the client can sample the classical description of a circuit $C_i$ of size $\mathrm{poly}(\log T, \lambda) = \mathrm{poly}(\lambda)$ and compute
\[
C_i \ket{0}_{\mathsf{Q}}\ket{0}_{\mathsf{W}} = \ket{\psi_i}_{\mathsf{Q}}\ket{0}_{\mathsf{W}}
\]
where $\ket{\psi_i}$ is the state sent from the client to the server in the $i$-th quantum communication round. In particular, the state sent in the $i$-th round is independent of the protocol transcript (but might depend on the client's internal randomness). We refer to an interactive protocol that satisfies such properties as being \emph{quantum-succinct and history-independent}. 
\begin{definition}[Quantum-Succinctness and History-Independence]\label{def:qshi}
A protocol with public input $x$ is quantum-succinct and history-independent if the client: (i) samples a uniform classical seed $s$ of fixed $\poly(\lambda)$ length before interaction; (ii) computes from $(s,x)$, in time $|x|\poly(\lambda)$, descriptions of circuits $C_i$ of fixed total size $\poly(\lambda)$, each preparing a message in fresh registers 
\[
C_i\ket{0}_{\mathsf{Q}_i}\ket{0}_{\mathsf{W}_i}
=\ket{\psi_i}_{\mathsf{Q}_i}\ket{0}_{\mathsf{W}_i},
\]
sends $\mathsf{Q}_i$ at a publicly pre-determined round, and discards $\mathsf{W}_i$; and (iii) is otherwise classical, with messages and acceptance determined by $s$, $x$, any private classical input, and the classical transcript.
\end{definition}
The protocol above satisfies this definition: Its preparation circuits use only the short sampled key and string lists and efficient PRF evaluation. These lists may be sampled directly, with only the remaining classical coins expanded from a PRG seed. Replacing those coins by PRG output changes QPT views only negligibly; the proofs below use uniform coins. Moreover, we say that a protocol is a \emph{claw-state correlation} protocol if it satisfies \Cref{lem:claw-correct} and \Cref{lem:claw-correlation}, which we prove in the following.
\begin{lemma}\label{lem:claw-correct}
    When both parties are honest, except with negligible probability, at the end of the protocol execution the server holds the state
    \[
    \bigotimes_{t\in\{0,1\}^{\log T}}\ket{\psi_{Y_{t,0},Y_{t,1}}}
    \]
    and the client holds $\{Y_{t,0},Y_{t,1}\}_{t\in\{0,1\}^{\log T}}$.
\end{lemma}
\begin{proof}
    Fix $t$. By correctness of the weak protocol, before gluing the server holds the tensor product of the $\lambda$ claws indexed by $\pi_i(t)$. After the CNOTs this state is
    \[
    2^{-\lambda/2}\!
    \sum_{b_1,\ldots,b_\lambda\in\{0,1\}}
    \ket{b_1,b_1\oplus b_2,\ldots,b_1\oplus b_\lambda}
    \bigotimes_{i=1}^{\lambda}\ket{y^{(i)}_{\pi_i(t),b_i}}.
    \]
    Measuring the last $\lambda-1$ leading qubits fixes $b_i=b_1\oplus e_{i,t}$ for $i\geq2$, leaving two equally weighted terms. After normalization and discarding the measured qubits, the remaining state is therefore
    \[
    \frac{\ket{0,Y_{t,0}}+\ket{1,Y_{t,1}}}{\sqrt2}
    =\ket{\psi_{Y_{t,0},Y_{t,1}}},
    \]
    with strings given by \Cref{eq:combined-endpoints}. The gluing operations act separately for each $t$, so the resulting states form the claimed tensor product, and the communicated measurement outcomes let the client compute the same strings. Finally, each weak copy rejects an honest server with probability at most $(\log T)2^{-n}$. By a union bound over the $\lambda$ copies, the abort probability is at most $\lambda (\log T)2^{-n}=\negl(\lambda)$.
\end{proof}
The following lemma establishes the security of the protocol.

\begin{lemma}\label{lem:claw-correlation}
    For every non-uniform QPT server and every $t^*\in\{0,1\}^{\log T}$, consider the experiment that runs the protocol, gives the server $\{Y_{t,0},Y_{t,1}\}_{t\neq t^*}$ after the client produces its output, and asks it to return $(Y_{t^*,0},Y_{t^*,1})$. Assuming the existence of quantum-secure one-way functions, the probability that the client does not abort and the server returns the correct pair is negligible.
\end{lemma}
\begin{proof}
    Fix any $t^*=t^*(\lambda)\in\{0,1\}^{\log T}$. We reduce recovery of its pair of strings to winning $\lambda$ sequential copies of the search game in \Cref{thm:single}. The private output of each game is its target pair if the check passes, and $\bot$ otherwise. The reduction forwards each weak-protocol interaction to the server, but stores the disclosed target index $u_i$ and strings for the other claws without revealing them.

    After all copies have finished, the reduction samples independent permutations uniformly subject to $\pi_i(t^*)=u_i$ and sends them to the server. Since the $u_i$ are independent uniform indices and were withheld during the weak interactions, this has exactly the joint distribution of the original protocol. Once the server sends its gluing bits, the stored strings determine every combined pair at $t\neq t^*$, so the reduction can supply all the required leakage.

    A correct guess of $(Y_{t^*,0},Y_{t^*,1})$ gives every raw target pair by splitting the strings into blocks and undoing the swaps specified in \Cref{eq:combined-endpoints}. Failed checks can be followed by dummy continuations, since those branches cannot win. Thus \Cref{thm:single,thm:sequential-repetition} give
    \[
    \Pr[\text{no abort and a correct target pair}]
    \leq 4^{-\lambda}+\negl(\lambda)=\negl(\lambda).
    \]
    Non-uniformity makes this bound uniform over target sequences. A union bound over the polynomially many targets also gives hardness of recovering any pair without being given the other strings.
\end{proof}

\subsection{Blind Delegation}
We recall the definition of a blind delegation protocol.
\begin{definition}[Blind Delegation]\label{def:blind}
A blind delegation protocol for a public quantum circuit, viewed as a map $Q:\mathsf{X}\mathsf{B}\to\mathsf{O}$, takes a private classical input $x$ from the client and a quantum input on $\mathsf{B}$ from the server. In an honest execution, the client obtains one-time-pad keys $(r,s)$. For every input state $\rho_{\mathsf{BR}}$ with arbitrary reference $\mathsf{R}$, let $\sigma_{\mathsf{OR}}$ be the server's output and reference after applying $\mathsf{Z}^s\mathsf{X}^r$ on $\mathsf{O}$ whenever the client does not abort, averaged over all randomness and measurement outcomes. Represent abort by an orthogonal flag in $\mathsf{O}$, and embed the ideal output in the non-abort subspace. Correctness requires
\[
\frac12\left\|\sigma_{\mathsf{OR}}-
(Q\otimes I_{\mathsf{R}})
\bigl(\ketbra{x}{x}_{\mathsf{X}}\otimes\rho_{\mathsf{BR}}\bigr)\right\|_1
\leq\negl(\lambda).
\]
For every QPT server $\mathcal A$, every $\rho_{\mathsf{BR}}$, and all equal-length inputs $x_0,x_1$, blindness requires
\[
\mathsf{View}_{\mathcal A}(\rho_{\mathsf{BR}},x_0)
\approx_c
\mathsf{View}_{\mathcal A}(\rho_{\mathsf{BR}},x_1),
\]
where the view includes the server's entire output, $\mathsf{R}$, and the abort flag, with the same public circuit $Q$ in both experiments. The factor $I_{\mathsf{R}}$ leaves the reference untouched.
\end{definition}
Using the protocol defined above and a result from \cite{BK25}, we can establish the existence of a quantum-succinct and history-independent blind delegation protocol.
\begin{lemma}\label{lem:blind}
    Assuming the existence of quantum-secure one-way functions, for any polynomial-size quantum circuit $Q$, there exists a quantum-succinct and history-independent (\Cref{def:qshi}) blind delegation protocol for $Q$. Moreover, its round complexity is bounded by $\mathrm{poly}(\lambda, \delta)$, where $\delta$ is the $\mathsf{T}$-depth of $Q$.
\end{lemma}
\begin{proof}

    Use a Clifford+$\mathsf{T}$ realization of $Q$, including any negligible approximation error in correctness. Let $T = \poly(\lambda)$ be a power-of-two upper bound on its number of $\mathsf{T}$ gates and let $\Pi_\claw$ be a quantum-succinct and history-independent claw-state correlation protocol with $T$ output claw states and round complexity $\poly(\lambda)$, i.e. it satisfies \cref{def:qshi}, \cref{lem:claw-correct}, and \cref{lem:claw-correlation}. Above we showed that $\Pi_\claw$ exists assuming quantum-secure one-way functions. Unused correlations may be discarded; if there are no $\mathsf{T}$ gates, omit this setup.

    We will use $\Pi_\claw$ to construct a blind delegation protocol $\Pi_\blind$ for $Q$ by applying two transformations from \cite{BK25}: (1) $\mathsf{BB84}$ correlations from claw-state correlations, and (2) blind delegation from $\mathsf{BB84}$ correlations.

    First, we specify the definition of a $\mathsf{BB84}$ correlation protocol $\Pi_\bbstate$. Let $T = \poly(\lambda)$ be the number of output correlations.

    \begin{itemize}
        \item Correctness: When both parties are honest, except with negligible probability, at the end of the protocol execution the server holds the state \[\bigotimes_{i \in [T]}\mathsf{H}^{\theta_i}\ket{x_i},\] and the client holds $\{\theta_i,x_i\}_{i \in [T]}$.
        \item Security: For every non-uniform QPT server and every $i^*=i^*(\lambda)\in[T]$, let $\mathrm{ok}$ denote that the client does not abort. On this branch give the server $\{\theta_j,x_j\}_{j\neq i^*}$ and let $\widehat\theta_{i^*}$ be its guess. We require
        \begin{equation}\label{eq:bb84-security}
        \left|\Pr[\mathrm{ok}\wedge\widehat\theta_{i^*}=\theta_{i^*}]
        -\tfrac12\Pr[\mathrm{ok}]\right|\leq\negl(\lambda).
        \end{equation}
        The view includes the abort flag and the server's full workspace. On abort the verifier samples the basis bits $\theta_i$ independently and uniformly at random.
    \end{itemize}

    For the first transformation, \cite[Theorem~4.7]{BK25} converts a claw-state generator into an oblivious state preparation (OSP) protocol. Apply this conversion in parallel to each of our $T$ claws. To meet its requirement that the two strings be distinct, set $L_{i,b}:=b\Vert Y_{i,b}$ and have the server copy its leading qubit into a fresh data qubit, obtaining $\frac1{\sqrt2}(\ket{0,L_{i,0}}+\ket{1,L_{i,1}})$. Recovering both strings recovers the original pair. For security of coordinate $i$, the reduction receives the other pairs of strings, simulates the other conversion clients, and computes their BB84 labels. The Goldreich--Levin reduction in that theorem then applies jointly with this leakage and the server's workspace: a noticeable advantage in the unconditional bound above gives a noticeable probability of recovering the target pair without abort, contradicting \cref{lem:claw-correlation}. Correctness follows from the cited conversion and \Cref{lem:claw-correct}. The transformation adds only classical communication and polynomially many rounds, so the resulting $\Pi_\bbstate$ remains quantum-succinct and history-independent.

    Generate this entire batch before sending any message depending on the client's input.
    For the second transformation, \cite[Lemma~6.10 and Theorem~6.11]{BK25} gives a protocol for blind delegation from any OSP. The BK25 protocol runs one OSP for each $\mathsf{T}$-gate, and here we will use the $i$-th OSP correlation $\mathsf{H}^{\theta_i}\ket{x_i}, (\theta_i,x_i)$ in place of the OSP protocol for the $i$-th $\mathsf{T}$-gate. In \cite{BK25}, the protocol for the $i$-th $\mathsf{T}$-gate requires the client to input a \emph{chosen} basis $\theta_i'$, while our OSP correlation yields a basis $\theta_i$ that is not controlled by the client. To remedy this, we simply have the client first send the bit $b_i = \theta_i \oplus \theta_i'$ to the server, who applies $\mathsf{H}^{b_i}$ to its state. After this, correctness follows directly from the arguments in \cite{BK25}. Further, note that the security of the OSP correlation protocol implies that the server's view on client input $\theta' = 0$ and $\theta' = 1$ are indistinguishable even given $\{\theta_j,x_j\}_{j \neq i}$.

    For blindness, switch the chosen bases to $0$ in reverse gate order, as in \cite[Theorem~6.11]{BK25}. At gate $k$, the reduction knows all labels except $(\theta_k,x_k)$: the earlier labels and transcript determine $\theta'_k$, while all later chosen bases have already been replaced by $0$. The BB84 guarantee makes $\theta_k$ indistinguishable from a fresh uniform bit jointly with the known labels and the server's state, so sending $\theta_k\oplus\theta'_k$ or $\theta_k$ gives indistinguishable views. Subsequent visible messages use only known labels; the unknown $x_k$ affects only private key updates. After polynomially many switches, the only input-dependent message is protected by the initial classical one-time pad, giving statistical blindness.

    Finally, as also observed in \cite{BLM26}, the client and server can participate in all OSPs for a given layer of $\mathsf{T}$-gates in parallel, with that layer's chosen bases fixed before its messages are sent. Thus, the round complexity of the blind delegation protocol is bounded by $\poly(\lambda,\delta)$. Moreover, since this transformation requires only classical communication, the blind delegation protocol remains quantum-succinct and history-independent.

\end{proof}

\subsection{Putting Things Together}

In \cite[Theorem~54]{BLM26} it is shown that a blind delegation protocol implies an argument for QMA.
\begin{lemma}\label{lem:qma}
Assuming a blind delegation protocol with round complexity $\poly(\lambda, \delta)$ with $\delta$ the $\mathsf{T}$-depth of the computed quantum circuit, there is an argument for QMA with completeness $1-\negl(\lambda)$, soundness at most a constant $s<1$, and a fixed $\mathrm{poly}(\lambda)$ number of rounds. 

Moreover, if the blind delegation protocol satisfies \Cref{def:qshi}, then so does the resulting interactive argument.
\end{lemma}
The moreover part of the above statement is not explicitly stated in \cite{BLM26}, but it is implicit in their proof because the protocol uses the blind delegation (see also \cite[Theorem~40]{BLM26}) as a black-box and everything else is classical communication. Next, we invoke the communication compression compiler from \cite{BLM26}.
\begin{lemma}\label{lem:compression}
Let $\Pi$ be an $r(\lambda)$-round interactive protocol satisfying \Cref{def:qshi}, with no private client input besides $s$, where $r$ is bounded by a fixed polynomial and the public input $x$ specifies the classical next-message and decision algorithms with description size $O(|x|)$. Assuming collapsing hash functions, for every $\varepsilon=1/\mathrm{poly}(\lambda)$ there is a compiled protocol $\tilde\Pi$ with negligible honest-completeness loss such that, for every QPT compiled prover $\tilde P^*$, there is a QPT original prover $P^*$ with
 \[
 \Pr[\tilde V^{\tilde P^*}\text{ accepts}]
 \leq\Pr[V^{P^*}\text{ accepts}]+O(\varepsilon)+\negl(\lambda).
 \]
Its communication is at most $r(\lambda)\poly(\lambda,1/\varepsilon)$ and its verifier time is at most $|x|r(\lambda)\poly(\lambda,1/\varepsilon)$.
\end{lemma}
\begin{proof}
    \cite[Theorem~18]{BLM26} states the same theorem for a protocol with a completely classical verifier, but we argue that the same holds for any interactive protocol satisfying \Cref{def:qshi}. First, recall that we can assume without loss of generality that the verifier's classical computation is deterministic once it samples a random seed $s\sim \{0,1\}^{\poly(\lambda)}$ prior to any interaction. Keep the short randomness for quantum preparation explicitly in $s$ and derive the remaining classical coins using a quantum-secure PRG. Note that by history-independence (\Cref{def:qshi}), this also means that the quantum states sent by the verifier are fully determined prior to any interaction.

    Let us now recall the compilation procedure for a classical verifier from \cite{BLM26}. For any round of interaction:
    \begin{itemize}
        \item The prover computes locally its next message of the protocol and sends a succinct commitment to it, together with a state-preserving succinct argument of knowledge for its pre-image.
        \item The prover and the verifier engage in a chosen-input committed secure function sampling (SFS) protocol, where the verifier's input is $s$ and the prover's inputs are the committed messages. At the end of the interaction, the prover receives the value of the verifier's next-message function.
\end{itemize}
At the end of the interaction, the verifier reveals $s$, and the prover supplies a succinct argument of knowledge proving that the committed messages form an accepting transcript. 

The verifier also sends the original quantum states at their prescribed positions. For soundness, keep $s$ and all external registers as an untouched reference. Following \cite[Definitions~11 and~19, Lemma~20, and Section~5.4]{BLM26}, extract each committed message using the state-preserving extractor and simulate the corresponding chosen-input SFS from the next-message value on the extracted transcript. These guarantees preserve the joint state up to the prescribed distinguishing error, including correlations between the prover's quantum registers and $s$; the chosen-input transformation uses fresh garbling randomness. At quantum-message positions, the reduction forwards the original verifier's register once, without knowing $s$ or cloning the state. Choose the inverse-polynomial per-call accuracies so that all extraction and simulation errors sum to $O(\varepsilon)$; the fixed polynomial bound on $r$ absorbs the resulting overhead.

After the last commitment, extract the final accepting-transcript witness. Collision resistance forces its messages to equal those already extracted, except with negligible probability. Hence compiled acceptance implies that the original verifier's classical predicate accepts the extracted transcript, up to the stated errors. The final seed reveal and proof cannot change those fixed messages and can be omitted for this upper bound. Thus an original prover forwards the extracted classical messages, uses the verifier's classical replies as SFS-simulator inputs, and forwards its quantum replies as above. Honest correctness of the subprotocols separately gives negligible completeness loss. Their efficiency and \Cref{def:qshi} give the claimed bounds.
\end{proof}

Combining \Cref{lem:claw-correlation,lem:blind,lem:qma,lem:compression}, along with the fact that collapsing hashes imply the existence of one-way functions, we obtain a succinct argument for QMA with constant soundness. Sequential repetition of \Cref{thm:main} yields our main result \Cref{thm:intro}.

\begin{theorem}\label{thm:main}
Assuming collapsing hash functions, there exists a constant $0<s<1$ and a succinct argument for QMA with completeness $1-\negl(\lambda)$, soundness error $s$, total communication bounded by $\mathrm{poly}(\lambda)$, and verifier runtime $|x|\cdot\mathrm{poly}(\lambda)$.
\end{theorem}
\begin{proof}
Let $s_0<1$ be the soundness in \Cref{lem:qma}. Choose the compilation accuracy so that \Cref{lem:compression} gives soundness at most $s=(1+s_0)/2<1$ for sufficiently large $\lambda$, while honest completeness remains $1-\negl(\lambda)$. The fixed round bound gives the claimed communication and verifier time. For \Cref{thm:intro}, repeat sequentially $\lambda$ times with fresh verifier randomness, accepting only if every copy accepts, and give the honest prover fresh witness copies. Completeness loss is still negligible by a union bound. Apply \Cref{thm:sequential-repetition} to the search game whose private target is $0$ on acceptance and $\bot$ otherwise, soundness is at most $s^\lambda+\negl(\lambda)$.
\end{proof}

\subsection*{Acknowledgments} 
JB is supported by the Air Force Office of Scientific Research under agreement number
FA9550261B239. GM is supported by the European Research Council through an ERC Starting Grant (Grant agreement No.~101077455, ObfusQation) and by the Deutsche Forschungsgemeinschaft (DFG, German Research Foundation) under Germany's Excellence Strategy - EXC 2092 CASA – 390781972. 

\phantomsection
\addcontentsline{toc}{section}{References}
\bibliographystyle{alpha}
\bibliography{refs}

@inproceedings{Kil92,
  author    = {Joe Kilian},
  title     = {A Note on Efficient Zero-Knowledge Proofs and Arguments (Extended Abstract)},
  booktitle = {Proceedings of the 24th Annual ACM Symposium on Theory of Computing (STOC)},
  pages     = {723--732},
  publisher = {ACM},
  year      = {1992},
  doi       = {10.1145/129712.129782},
  note      = {\url{https://doi.org/10.1145/129712.129782}},
  url       = {https://doi.org/10.1145/129712.129782}
}

@article{ALM+98,
  author  = {Sanjeev Arora and Carsten Lund and Rajeev Motwani and Madhu Sudan and Mario Szegedy},
  title   = {Proof Verification and the Hardness of Approximation Problems},
  journal = {Journal of the ACM},
  volume  = {45},
  number  = {3},
  pages   = {501--555},
  year    = {1998},
  doi     = {10.1145/278298.278306},
  note      = {\url{https://doi.org/10.1145/278298.278306}},
  url     = {https://doi.org/10.1145/278298.278306}
}

@article{AS98,
  author  = {Sanjeev Arora and Shmuel Safra},
  title   = {Probabilistic Checking of Proofs: A New Characterization of {NP}},
  journal = {Journal of the ACM},
  volume  = {45},
  number  = {1},
  pages   = {70--122},
  year    = {1998},
  doi     = {10.1145/273865.273901},
  note      = {\url{https://doi.org/10.1145/273865.273901}},
  url     = {https://doi.org/10.1145/273865.273901}
}

@inproceedings{BKL+22,
  author    = {James Bartusek and {Tauman Kalai}, Yael and Alex Lombardi and Fermi Ma and Giulio Malavolta and Vinod Vaikuntanathan and Thomas Vidick and Lisa Yang},
  title     = {Succinct Classical Verification of Quantum Computation},
  booktitle = {Advances in Cryptology---CRYPTO 2022},
  series    = {Lecture Notes in Computer Science},
  volume    = {13508},
  pages     = {195--211},
  publisher = {Springer},
  year      = {2022},
  doi       = {10.1007/978-3-031-15979-4_7},
  note      = {\url{https://arxiv.org/abs/2206.14929}},
  url       = {https://arxiv.org/abs/2206.14929}
}

@inproceedings{MNZ24,
  author    = {Tony Metger and Anand Natarajan and Tina Zhang},
  title     = {Succinct Arguments for {QMA} from Standard Assumptions via Compiled Nonlocal Games},
  booktitle = {2024 IEEE 65th Annual Symposium on Foundations of Computer Science (FOCS)},
  pages     = {1193--1201},
  publisher = {IEEE},
  year      = {2024},
  doi       = {10.1109/FOCS61266.2024.00078},
  note      = {\url{https://arxiv.org/abs/2404.19754}},
  url       = {https://arxiv.org/abs/2404.19754}
}

@inproceedings{GTKNV25,
  author    = {Sam Gunn and {Tauman Kalai}, Yael and Anand Natarajan and Vill{\'a}nyi, {\'A}gi},
  title     = {Classical Commitments to Quantum States},
  booktitle = {Proceedings of the 57th Annual ACM Symposium on Theory of Computing (STOC)},
  pages     = {234--244},
  publisher = {ACM},
  year      = {2025},
  doi       = {10.1145/3717823.3718264},
  note      = {\url{https://arxiv.org/abs/2404.14438}},
  url       = {https://arxiv.org/abs/2404.14438}
}

@inproceedings{BLM26,
  author    = {James Bartusek and Jiahui Liu and Giulio Malavolta},
  title     = {A Modular Approach to Succinct Arguments for {QMA}},
  booktitle = {Advances in Cryptology---EUROCRYPT 2026},
  series    = {Lecture Notes in Computer Science},
  volume    = {16547},
  pages     = {446--474},
  publisher = {Springer},
  year      = {2026},
  doi       = {10.1007/978-3-032-25336-1_16},
  note      = {Full version: \url{https://arxiv.org/abs/2606.10408v1}},
  url       = {https://arxiv.org/abs/2606.10408v1}
}

@inproceedings{BK25,
  author    = {James Bartusek and Dakshita Khurana},
  title     = {On the Power of Oblivious State Preparation},
  booktitle = {Advances in Cryptology---CRYPTO 2025},
  series    = {Lecture Notes in Computer Science},
  volume    = {16001},
  pages     = {575--607},
  publisher = {Springer},
  year      = {2025},
  doi       = {10.1007/978-3-032-01878-6_19},
  note      = {Full version: \url{https://arxiv.org/abs/2411.04234v1}},
  url       = {https://arxiv.org/abs/2411.04234v1}
}

@inproceedings{BMM25,
  author    = {Pedro Branco and Giulio Malavolta and Zayd Maradni},
  title     = {Fully-Homomorphic Encryption from Lattice Isomorphism},
  booktitle = {Theory of Cryptography---TCC 2025},
  series    = {Lecture Notes in Computer Science},
  volume    = {16268},
  pages     = {220--252},
  publisher = {Springer},
  year      = {2026},
  doi       = {10.1007/978-3-032-12287-2_8},
  note      = {First published online in December 2025. \url{https://eprint.iacr.org/2025/993}},
  url       = {https://eprint.iacr.org/2025/993}
}

@inproceedings{CMSZ22,
  author    = {Alessandro Chiesa and Fermi Ma and Nicholas Spooner and Mark Zhandry},
  title     = {Post-Quantum Succinct Arguments: Breaking the Quantum Rewinding Barrier},
  booktitle = {2021 IEEE 62nd Annual Symposium on Foundations of Computer Science (FOCS)},
  pages     = {49--58},
  publisher = {IEEE},
  year      = {2022},
  doi       = {10.1109/FOCS52979.2021.00014},
  note      = {\url{https://arxiv.org/abs/2103.08140}},
  url       = {https://arxiv.org/abs/2103.08140}
}

@inproceedings{GJMZ23,
  author    = {Sam Gunn and Nathan Ju and Fermi Ma and Mark Zhandry},
  title     = {Commitments to Quantum States},
  booktitle = {Proceedings of the 55th Annual ACM Symposium on Theory of Computing (STOC)},
  pages     = {1579--1588},
  publisher = {ACM},
  year      = {2023},
  doi       = {10.1145/3564246.3585198},
  note      = {\url{https://arxiv.org/abs/2210.05138}},
  url       = {https://arxiv.org/abs/2210.05138}
}

@inproceedings{KLVY23,
  author    = {Yael Kalai and Alex Lombardi and Vinod Vaikuntanathan and Lisa Yang},
  title     = {Quantum Advantage from Any Non-Local Game},
  booktitle = {Proceedings of the 55th Annual ACM Symposium on Theory of Computing (STOC)},
  pages     = {1617--1628},
  publisher = {ACM},
  year      = {2023},
  doi       = {10.1145/3564246.3585164},
  note      = {\url{https://arxiv.org/abs/2203.15877}},
  url       = {https://arxiv.org/abs/2203.15877}
}

@inproceedings{NZ23,
  author    = {Anand Natarajan and Tina Zhang},
  title     = {Bounding the Quantum Value of Compiled Nonlocal Games: From {CHSH} to {BQP} Verification},
  booktitle = {2023 IEEE 64th Annual Symposium on Foundations of Computer Science (FOCS)},
  pages     = {1342--1348},
  publisher = {IEEE},
  year      = {2023},
  doi       = {10.1109/FOCS57990.2023.00081},
  note      = {\url{https://arxiv.org/abs/2303.01545v2}},
  url       = {https://arxiv.org/abs/2303.01545v2}
}

@inproceedings{BKMSW25,
  author    = {Kaniuar Bacho and Alexander Kulpe and Giulio Malavolta and Simon Schmidt and Michael Walter},
  title     = {Compiled Nonlocal Games from Any Trapdoor Claw-Free Function},
  booktitle = {Advances in Cryptology---CRYPTO 2025},
  series    = {Lecture Notes in Computer Science},
  volume    = {16001},
  pages     = {642--673},
  publisher = {Springer},
  year      = {2025},
  doi       = {10.1007/978-3-032-01878-6_21},
  note      = {\url{https://eprint.iacr.org/2024/1829}},
  url       = {https://eprint.iacr.org/2024/1829}
}

@inproceedings{Zha21,
  author    = {Jiayu Zhang},
  title     = {Succinct Blind Quantum Computation Using a Random Oracle},
  booktitle = {Proceedings of the 53rd Annual ACM SIGACT Symposium on Theory of Computing (STOC)},
  pages     = {1370--1383},
  publisher = {ACM},
  year      = {2021},
  doi       = {10.1145/3406325.3451082},
  note      = {\url{https://arxiv.org/abs/2004.12621}},
  url       = {https://arxiv.org/abs/2004.12621}
}

@inproceedings{GSLW19,
  author    = {Andr{\'a}s Gily{\'e}n and Yuan Su and Guang Hao Low and Nathan Wiebe},
  title     = {Quantum Singular Value Transformation and Beyond: Exponential Improvements for Quantum Matrix Arithmetics},
  booktitle = {Proceedings of the 51st Annual ACM SIGACT Symposium on Theory of Computing (STOC)},
  pages     = {193--204},
  publisher = {ACM},
  year      = {2019},
  doi       = {10.1145/3313276.3316366},
  note      = {Full version: \url{https://arxiv.org/abs/1806.01838v1}},
  url       = {https://arxiv.org/abs/1806.01838v1}
}

@inproceedings{Zhandry12,
  author    = {Mark Zhandry},
  title     = {How to Construct Quantum Random Functions},
  booktitle = {2012 IEEE 53rd Annual Symposium on Foundations of Computer Science (FOCS)},
  pages     = {679--687},
  publisher = {IEEE},
  year      = {2012},
  doi       = {10.1109/FOCS.2012.37},
  note      = {\url{https://eprint.iacr.org/2012/182}},
  url       = {https://eprint.iacr.org/2012/182}
}

@inproceedings{Unr16,
  author    = {Dominique Unruh},
  title     = {Computationally Binding Quantum Commitments},
  booktitle = {Advances in Cryptology---EUROCRYPT 2016},
  series    = {Lecture Notes in Computer Science},
  volume    = {9666},
  pages     = {497--527},
  publisher = {Springer},
  year      = {2016},
  doi       = {10.1007/978-3-662-49896-5_18},
  note      = {\url{https://eprint.iacr.org/2015/361}},
  url       = {https://eprint.iacr.org/2015/361}
}

@misc{cryptoeprint:2026/2099,
      author = {Alessandro Chiesa and Zihan Hu},
      title = {Succinct Arguments for {QMA} in the Quantum Random Oracle Model},
      howpublished = {Cryptology {ePrint} Archive, Paper 2026/2099},
      year = {2026},
      url = {https://eprint.iacr.org/2026/2099}
}

@InProceedings{sun_et_al:LIPIcs.CCC.2026.4,
  author =	{Sun, Baocheng and Vidick, Thomas},
  title =	{{Probabilistically Checking Quantum Proofs, with Interaction}},
  booktitle =	{41st Computational Complexity Conference (CCC 2026)},
  pages =	{4:1--4:49},
  series =	{Leibniz International Proceedings in Informatics (LIPIcs)},
  ISBN =	{978-3-95977-437-6},
  ISSN =	{1868-8969},
  year =	{2026},
  volume =	{383},
  editor =	{Moshkovitz, Dana},
  publisher =	{Schloss Dagstuhl -- Leibniz-Zentrum f{\"u}r Informatik},
  address =	{Dagstuhl, Germany},
  URL =		{https://drops.dagstuhl.de/entities/document/10.4230/LIPIcs.CCC.2026.4},
  URN =		{urn:nbn:de:0030-drops-270463},
  doi =		{10.4230/LIPIcs.CCC.2026.4}
}

\end{document}